\documentclass[12pt]{article}
\usepackage{amsmath,amssymb,amsthm,bm}
\usepackage{mathbbol,bbm}
\usepackage[final]{showkeys}
\usepackage{verbatim}
\newcommand{\Cx}{{\mathbb C}}
\newcommand{\Ir}{\mathbb{Z}}
\newcommand{\Nl}{\mathbb{N}}

\newcommand{\Rl}{{\mathbb R}}

\newcommand{\idty}{\mathbb{1}}

\DeclareMathOperator{\ran}{ran}

\DeclareMathOperator{\spa}{span}
\DeclareMathOperator*{\tr}{Tr}
\newcommand{\<}{\langle}
\renewcommand{\>}{\rangle}

\providecommand{\norm}[1]{\lVert#1\rVert}

\renewcommand{\c}[1]{\mathcal{#1}}
\newcommand{\g}[1]{\mathfrak{#1}}
\newcommand{\s}[1]{\mathsf{#1}}
\renewcommand{\r}[1]{\mathrm{#1}}

\newtheorem{lemma}{Lemma}
\newtheorem{proposition}{Proposition}

\begin{document}

\begin{center}
{\LARGE Correlations in Free Fermionic States} \\[12pt]
M.~Fannes and J.~Wouters \\[6pt]
Instituut voor Theoretische Fysica \\
K.U.Leuven, B-3001 Heverlee, Belgium
\end{center}

\medskip\noindent
\textbf{Abstract}

\noindent 
We study correlations in a bipartite, Fermionic, free state in terms of perturbations induced by one party on the other. In particular, we show that all so conditioned free states can be modelled by an auxiliary Fermionic system and a suitable completely positive map.

\section{Introduction}
\label{s1}

The rich structure of multi-partite quantum states that arises from the interplay between probability and locality leads to many interesting concepts, features, and problems~\cite{NC}. To determine whether a bipartite state is entangled or not or to understand the nature of entanglement in higher dimensional systems turn out to be very hard problems~\cite{HHHH}. In this paper we consider the influence of perturbations of one party on the other. The states that arise in this way could be called conditional states even if the classical notion of conditioning cannot be extended to the quantum~\cite{T}. More precisely, we work out the structure of conditional state spaces for fermionic systems, imposing additionally Gaussianity.  

We consider here mostly finite dimensional algebras of observables $\c A$, these can always be taken to be unital sub-algebras of some complex matrix algebra $\c M$, closed under Hermitian conjugation. Such algebras are direct sums of full matrix algebras and therefore encompass both classical systems with finite state spaces and fully quantum systems with a finite number of accessible levels. The space of complex linear functionals on $\c A$ is denoted by $\c A^*$ and the pairing between a functional $\varphi$ and an observable $A$ by $\varphi \cdot A$. Because of the finite dimensions $(\c A^*)^* = \c A$. The state space of $\c A$ is the convex subset $\c S(\c A)$ of normalized, positive, linear functionals. The term state therefore means expectation functional rather than wave function as in standard quantum mechanics.

We now consider bipartite systems. The observables of both parties form algebras $\c A_1$ and $\c A_2$ and of the composite system $\c A_{12} = \c A_1 \otimes \c A_2$. Product states are of the form 
\begin{equation}
\< A_1 \otimes A_2 \>_{12} = \< A_1 \>_1\, \< A_2 \>_2,
\end{equation} 
they describe statistical independence. Generally, subsystems will be correlated and this is encoded in conditional state spaces
\begin{align}
&\c S_1 := \bigl\{ A_1 \mapsto \< A_1 \otimes A_2 \>_{12} \,\mid\, A_2 \in \c A_2^+,\ \< \idty_1 \otimes A_2 \>_{12} = 1 \bigr\} \enskip\text{and} \\
&\c S_2 := \bigl\{ A_2 \mapsto \< A_1 \otimes A_2 \>_{12} \,\mid\, A_1 \in \c A_1^+,\ \< A_1 \otimes \idty_2 \>_{12} = 1 \bigr\}.
\end{align}  
$\c S_i$ is a compact convex subset of $\c S(\c A_i)$.
 
We may also consider the linear spaces of functionals
\begin{align}
&\c V_1 := \bigl\{ A_1 \mapsto \< A_1 \otimes A_2 \>_{12} \,\mid\, A_2 \in \c A_2 \bigr\} \enskip\text{and} \\
&\c V_2 := \bigl\{ A_2 \mapsto \< A_1 \otimes A_2 \>_{12} \,\mid\, A_1 \in \c A_1 \bigr\}.
\end{align}  
As any element in a C*-algebra is a linear combination of at most four positive elements $\c V_i$ is spanned by $\c S_i$. Mostly, $\c S_i$ is a proper subset of the space of positive normalized functionals in $\c V_i$. 

A state $\<\ \>_{12}$ of a composite system is a linear map from $\c A_2$ to $\c A_1^*$
\begin{equation}
S: A_2 \in \c A_2 \mapsto \bigl( A_1 \in \c A_1 \mapsto \< A_1 \otimes A_2\> \bigr) \in \c A_1^*.
\end{equation} 
This map is, moreover, positive. The transposed map $S^{\s T}$ from $\c A_1$ to $\c A_2^*$ 
\begin{equation}
S^{\s T}(A_1) \cdot A_2 = S(A_2) \cdot A_1,\enskip A_i \in \c A_i
\end{equation}
simply swaps the parties. As the rank of a map and its transpose are equal 
\begin{equation}
\dim\bigl( \c V_1 \bigr) = \dim\bigl( \c V_2 \bigr) =: n.
\end{equation}
The natural number $n$ is the correlation dimension of $\<\ \>_{12}$. 

The conditional state $A_1 \mapsto \< A_1 \otimes A_2 \>$ on $\c A_1$ can now be written in the form
\begin{equation}
\< A_1 \otimes A_2 \>_{12} = S^{\s T}(A_1) \cdot A_2 = S^{\s T}(A_1) \cdot B
\label{deco}
\end{equation}
for a suitably chosen $B$ from $S^{\s T}(\c A_1)^*$, i.e., we have modelled the conditional states on $\c A_1$ by a $n$-dimensional space. However, $B$ does not have to be positive which makes~(\ref{deco}) not very useful. 

Consider a bipartite system with fully quantum parties, i.e., $\c A_i = \c M_i$ where $\c M_i$ is a full matrix algebra of dimension $d_i$. The general finite dimensional situation can be handled by decomposition in a direct sum of full matrix algebras. A state of the composite system is given by a density matrix $\rho_{12}$ of dimension $d_1 d_2$
\begin{equation}
\< A_{12} \>_{12} = \tr \bigl( \rho_{12} A_{12} \bigr),\enskip A_{12} \in \c M_1 \otimes \c M_2.
\end{equation}  
Let $d_3$ be the dimension of the range of $\rho$, then the GNS-construction yields an essentially unique normalized vector $\Omega \in \Cx^{d_1} \otimes \Cx^{d_2} \otimes \Cx^{d_3}$ such that
\begin{equation}
\< A_{12} \>_{12} = \< \Omega \,,\, A_{12} \otimes \idty_3\, \Omega \>,\enskip A_{12} \in \c M_1 \otimes \c M_2.
\end{equation}
We now perform the Schmidt decomposition of $\Omega$ with respect to $\c M_2$ and $\c M_1 \otimes \c M_3$
\begin{equation}
\Omega = \sum_{j=1}^m c_j \, \Omega_{2\,j} \otimes \Omega_{13\, j}.
\end{equation}
Here $c_j > 0$ and $\{ \Omega_{2\,j} \}$ and $\{ \Omega_{13\, j} \}$ are orthonormal families in $\Cx^{d_2}$ and $\Cx^{d_1} \otimes \Cx^{d_3}$. The conditional states on $\c M_1$ are convex combinations of conditional states defined by a rank one operator in $\c M_2$. These are of the form
\begin{equation}
\begin{split}
A_1 
&\mapsto \< \Omega \,,\, A_1 \otimes |\eta\>\<\eta| \otimes \idty_3\, \Omega\> \\
&= \sum_{k,\ell = 1}^m \< \Omega_{2\,k} \,,\, \eta \>\, \< \eta \,,\, \Omega_{2\, \ell} \>\, \< \Omega_{13\,k} \,,\, A_1 \otimes \idty_3\, \Omega_{13\, \ell} \> \\
&= \< \xi \,,\, A_1 \otimes \idty_3\, \xi \>.
\end{split}
\end{equation}
Here $\xi$ is a normalized vector in $\spa\bigl( \{ \Omega_{13\, j} \} \bigr)$. Moreover, any normalized $\xi$ can be reached by an appropriate choice of $\eta$. Therefore, the conditional states are of the form
\begin{equation}
A_1 \mapsto \< \xi \,,\, A_1 \otimes \idty_3\, \xi\>,\enskip \xi \in \spa\bigl( \{ \Omega_{13\, j} \} \bigr). 
\end{equation}
Picking an isometry $V$ from $\spa\bigl( \{ \Omega_{13\, j} \} \bigr)$ to $\Cx^m$ we obtain the following manifestly positive model for the conditional states
\begin{equation}
A_1 \mapsto \tr \bigl( \rho\, \Gamma(A_1) \bigr),\enskip \rho \text{ density matrix on } \Cx^m 
\label{deco2}
\end{equation}
and
\begin{equation}
\Gamma(A_1) = V\, A_1 \otimes \idty_3\, V^*.
\end{equation}
The map $\Gamma$ is completely positive and identity preserving. This description of conditional states fits in the general setting of generalized subsystems of~\cite{AFP}.

For a pure state $\<\ \>_{12}$ on $\c M_1 \otimes \c M_2$ defined by a normalized vector $\Omega_{12}$ the forms~(\ref{deco}) and~(\ref{deco2}) are actually very similar. We can identify the dual of $\c M$ with $\c M$ and use the pairing 
\begin{equation}
\varphi \cdot A = \tr (\varphi A).
\end{equation}
Writing the Schmidt decomposition
\begin{equation}
\Omega_{12} = \sum_{i=1}^p r_i^{\frac{1}{2}}\, e_i \otimes f_i,\enskip r_i > 0,
\end{equation}
we easily compute for $A_1 \in \c M_1$
\begin{equation}
S^{\s T}(A_1) = \sum_{k,\ell=1}^p r_k^{\frac{1}{2}}\, r_\ell^{\frac{1}{2}}\, \< e_k \,,\, A_1 e_\ell \>\, |f_\ell\>\<f_k|.
\end{equation}

It is not hard to verify that
\begin{equation}
\c S_2 = \bigl\{ S^{\s T}(A_1) \bigm| A_1 \ge 0 \enskip\text{and}\enskip \< A_1 \otimes \idty_2 \>_{12} = 1 \bigr\}
\end{equation}
is affinely isomorphic to the state space of the $p$-dimensional complex matrices $\c M_p$.
We see therefore that the correlation dimension $n$ is $p^2$. A general conditional state on $\c A_1$ 
\begin{equation}
A_1 \mapsto \< A_1 \otimes A_2 \>_{12},\enskip A_2 \ge 0 \enskip\text{and}\enskip \< \idty_1 \otimes A_2 \>_{12} = 1
\end{equation}
can then be written as
\begin{equation}
A_1 \mapsto \tr \bigl( S^{\s T}( A_1 ) B \bigr) 
\label{decopure}
\end{equation}
for a suitable $B \in \bigl( S^{\s T}(\c A_1) \bigr)^*$. As $S^{\s T}(\c A_1)$ is the full state space of the $n$-dimensional matrices we must have that 
\begin{equation}
B \ge 0 \enskip\text{and}\enskip \tr \bigl( S^{\s T}(\idty_1) B \bigr) = 1.
\end{equation}
This means that~(\ref{decopure}) is manifestly positive. To obtain the equivalence with the form~(\ref{deco2}) we use the transposition with respect to the basis $\{f_j\}$ of $\Cx^p$:
\begin{equation}
\begin{split}
\tr \bigl( S^{\s T}(A_1) B \bigr) 
&= \tr \bigl( B^{\s T} \bigl(  S^{\s T}(A_1) \bigr)^{\s T} \bigr) = \tr \bigl( B^{\s T} \Lambda(A_1) \bigr) \\
&= \tr \Bigl( \Lambda(\idty_1)^{\frac{1}{2}} B^{\s T} \Lambda(\idty_1)^{\frac{1}{2}} \Gamma(A_1) \Bigr)
\end{split}
\end{equation} 
with
\begin{align}
&\Lambda(A_1) = \sum_{k,\ell = 1}^p r_k^{\frac{1}{2}} r_\ell^{\frac{1}{2}} \< e_k \,,\, A_1 e_\ell \>\, |f_k\> \<f_\ell| \enskip\text{and} \\
&\Gamma(A_1) = \Lambda(\idty)^{-\frac{1}{2}} \Lambda(A_1) \Lambda(\idty_1)^{-\frac{1}{2}}.
\end{align}

An extension of these ideas was made in the context of translation invariant states on quantum spin chains. A family of states called finitely correlated states or also matrix product states was studied~\cite{FNW}. The pure ones turn out to be ground states of VBS-models and they are also useful as trial states in numerical computations~\cite{VC}. To construct such a state on a quantum spin chain $\otimes^\Ir \c A$ (where the single site algebra $\c A$ is typically a matrix algebra) an auxiliary finite dimensional algebra $\c B$ is introduced together with a unity preserving completely positive map
\begin{equation}
\s E: \c A \otimes \c B \to \c B.
\end{equation} 
Introducing the super-operators 
\begin{equation}
\s E_A: \c B \to \c B: B \mapsto \s E(A \otimes B),\enskip A \in \c A 
\end{equation}
and assuming for simplicity that repeated actions of $\s E_{\c A}$ on $\idty_{\c B}$ span the whole of $\c B$ and that $\idty_{\c B}$ is the unique eigenvector of $\s E_{\idty}$ with eigenvalue one, there exists a unique state $\rho$ on $\c B$ that satisfies
\begin{equation}
\rho(B) = \rho \bigl(\s E_{\idty}(B) \bigr),\enskip B \in \c B.
\end{equation}
The restrictions of the chain state $\omega$ to subsets of contiguous points are 
\begin{equation}
\omega(A_m \otimes A_{m+1} \otimes \cdots \otimes A_n) = \rho\Bigl( \s E_{A_m} \circ \s E_{A_{m+1}} \circ \cdots \circ \s E_{A_n} \bigl( \idty_{\c B} \bigr) \Bigr).
\end{equation}
The conditional states on the right half-chain $\otimes^{\Nl_0} \c A$ are then modelled by
\begin{equation}
X \mapsto \sigma(\Gamma(X))
\end{equation} 
where $\sigma$ is an arbitrary state on $\c B$ and
\begin{equation}
\Gamma: \otimes^{\Nl_0} \c A \to \c B: \Gamma(A_1 \otimes A_2 \otimes \cdots \otimes A_n) = \s E_{A_1} \circ \s E_{A_2} \circ \cdots \circ \s E_{A_n} (\idty).
\end{equation}

\section{Free fermionic states}
\label{s2}

Quantum states are mostly indirectly given, typically as ground or equilibrium states for a given interaction. Bosonic or fermionic free, quasi-free, Gaussian, or determinantal states are an exception, their two-point expectations are specified and the state is computed on general elements by applying a simple combinatorial rule. We shall restrict our attention to fermionic systems and compute conditional states within the free context. A general reference to this section is~\cite{BR}.

The CAR-algebra $\c A(\c H)$ --- CAR stands for canonical anti-commutation relations --- with one mode Hilbert space $\c H$ is the C*-algebra generated by an identity $\idty$ and by creation and annihilation operators $a^*$ and $a$ that satisfy
\begin{align}
&\varphi \in \c H \mapsto a^*(\varphi)\ \text{is complex linear} \\
&\{ a(\varphi) \,,\, a(\psi) \} = 0 \enskip\text{and}\enskip \{ a(\varphi) \,,\, a^*(\psi) \} = \< \varphi \,,\, \psi \>\, \idty. \label{CAR}
\end{align}   
An orthogonal decomposition $\c H = \c H_1 \oplus \c H_2$ turns $\c A(\c H)$ into a composite system with parties $\c A(\c H_i)$ up to a minor modification: $\c A(\c H_i)$ sits as a graded tensor factor in $\c A(\c H)$ through the natural identification $a^*(\varphi_i) \mapsto a^*(\varphi_i \oplus 0)$. This is due to the fact that odd elements in $\c A(\c H_1 \oplus 0)$ anti-commute with odd elements in $\c A(0 \oplus \c H_2)$. To simplify notation we shall often write $a^*(\varphi)$ instead of $a^*(\varphi \oplus 0)$. 

There is a representation from U(1) in the group $\{ \alpha_z \mid z \in \r{U(1)}\}$ of gauge automorphisms of $\c A(\c H)$ 
\begin{equation}
z \in \r{U(1)} \mapsto \alpha_z \enskip\text{with}\enskip \alpha_z(a^*(\varphi)) = z a^*(\varphi).
\end{equation}
It's fixed point algebra is the GICAR-algebra --- gauge-invariant CAR ---, it is generated as a linear space by monomials in creation and annihilation operators of the form $a^*(\varphi_1) \cdots a^*(\varphi_n) a(\psi_n) \cdots a(\psi_1)$.

A gauge-invariant free state $\omega_Q$ on $\c A(\c H)$ is determined by a symbol which is a linear operator $Q$ on $\c H$ satisfying $0 \le Q \le \idty$. The $\omega_Q$-expectations of all monomials vanish except for
\begin{equation}
\omega_Q \bigl( a^*(\varphi_1) a^*(\varphi_2) \cdots a^*(\varphi_n) a(\psi_n) \cdots a(\psi_2) a(\psi_1) \bigr) = 
\det\Bigl( \bigl[ \bigl\< \psi_k \,,\, Q\, \varphi_\ell \bigr\> \bigr] \Bigr).
\end{equation}

A different approach will prove useful here, see~\cite{D,DFP} for more details. The second quantization map 
\begin{equation}
\Gamma: \c T_1(\c H) \to \c A(\c H): \Gamma(A) := \sum_{k,\ell} \< e_k \,,\, A\, e_\ell\>\, a^*(e_k) a(e_\ell)
\end{equation}
takes a trace class operator $A \in \c T_1(\c H)$ to an element $\Gamma(A)$ in $\c A(\c H)$ that is independent of the chosen orthonormal basis $\{e_i\}$ of $\c H$. This map is complex-linear, continuous, and satisfies 
\begin{equation}
\frac{1}{2}\, \norm A_1 \le \norm{\Gamma(A)} \le \norm A_1.
\end{equation}
It is, moreover, completely positive and for a positive $A \in \c T_1(\c H)$
\begin{equation}
\norm{\Gamma(A)} = \tr A.
\end{equation}

In~\cite{DFP} a map $\r E$ from the Fredholm operators $\idty + \c T_1(\c H)$ to $\c A(\c H)$ is considered that satisfies
\begin{align}
&\r E(X) \r E(Y) = \r E(XY), \\
&\r E(X)^* = \r E(X^*),\enskip \text{and} \\
&\r E(\exp A) = \exp\bigl( \Gamma(A) \bigr),\enskip A \in \c T_1(\c H). 
\end{align}
This map obeys for positive trace-class $A$ the bounds
\begin{align}
&1 + \norm A_1 \le \norm{\r E(\idty + A)} \le \exp \bigl( \norm A_1 \bigr) \enskip\text{and} \\ 
&\norm{\r E(\idty + A) - \idty} \le \exp \bigl( \norm A_1 \bigr) - 1.
\end{align}
A gauge-invariant free state $\omega_Q$ can then be characterized by 
\begin{equation}
\label{gi-free}
\omega_Q(\r E(X)) = \det(\idty - Q + QX),\enskip X \in \idty + \c T_1(\c H).
\end{equation}

A state $\omega$ on $\c A(\c H)$ is even if it vanishes on all monomials in creation and annihilation operators with an odd number of factors. Gauge-invariant states are automatically even. If $\omega_i$ is an even state on $\c A(\c H_i)$ for $i=1,2$, then there exists a unique state $\omega_1 \wedge \omega_2$ on $\c A(\c H_1 \oplus \c H_2)$ such that
\begin{equation}
(\omega_1 \wedge \omega_2)(X_1 X_2) = \omega_1(X_1)\, \omega_2(X_2),\enskip X_i \in \c A(\c H_i).
\end{equation}
A symbol $Q$ induces an orthogonal decomposition
\begin{equation}
\c H = \c H_0 \oplus \tilde{\c H} \oplus \c H_{\idty}
\end{equation}
where
\begin{equation}
\c H_0 = \ker(Q) \enskip\text{and}\enskip \c H_{\idty} = \ker(\idty - Q)
\end{equation}
and $\omega_Q$ factorizes into
\begin{equation}
\omega_Q = \omega_0 \wedge \omega_{\tilde Q} \wedge \omega_{\idty} \enskip\text{with}\enskip \tilde Q = Q \bigr|_{\c H}.
\end{equation}
The states $\omega_0$ on $\c A(\c H_0)$ and $\omega_{\idty}$ on $\c A(\c H_{\idty})$ are pure, they are Fock and anti-Fock states.
 
We now consider a free state on a bipartite fermionic system $\c A(\c H_1 \oplus \c H_2)$ defined by a symbol $Q$ with block matrix structure
\begin{equation}
Q = \begin{bmatrix} A &B \\ B^* &C \end{bmatrix}.
\label{sym}
\end{equation}
The aim is to characterize all free states on $\c A(\c H_1)$ that arise as conditional states. More precisely, to characterize
\begin{equation}
\begin{split}
S_1^{\text{free}} = \Bigl\{ \omega_{\tilde A} \,\Bigm|\, 
&\omega_{\tilde A} \text{ is a free state on } \c A(\c H_1) \text{ and} \\
&\exists \text{ a gauge-invariant } Y \in \c A(\c H_2) \text{ such that} \\
&\omega_{\tilde A}(X) = \omega_Q(XY),\ X \in \c A(\c H_1) \Bigr\}.
\end{split}
\label{cond}
\end{equation}

From the positivity conditions $Q \ge 0$ and $Q \le \idty$ of the symbol~(\ref{sym}) on $\c H_1 \oplus \c H_2$ it immediately follows that
\begin{equation}
B\, \ker(C) = B\, \ker(\idty-C) = 0.
\end{equation}
This implies that the sub-algebras $\c A(\ker(C))$ and $\c A(\ker(\idty-C))$ of $\c A(\c H_2)$ are irrelevant for computing the conditional states~(\ref{cond}). There is therefore no loss in generality to assume that the kernels of $C$ and $\idty - C$ are trivial. In this case the positivity conditions can be restated as
\begin{equation}
0 < C < \idty, \enskip B C^{-1} B^* \le A \enskip\text{and}\enskip B (\idty - C)^{-1} B^* \le \idty - A.
\label{pos}
\end{equation}
In these inequalities, even if $C^{-1}$ or $(\idty - C)^{-1}$ are unbounded, $B C^{-1} B^*$ and $B (\idty - C)^{-1} B^*$ extend to bounded operators on $\c H_1$. The positivity conditions~(\ref{pos}) can be recast into  
\begin{equation}
0 \le A \le \idty \enskip\text{and}\enskip 0 \le C \le \idty 
\label{pos21}
\end{equation}
and there exist operators 
\begin{equation}
\begin{split}
&D_i : \c H_2 \to \c H_1,\enskip \norm{D_i} \le 1,\ i = 1,2 \enskip\text{such that} \\ 
&B = A^{\frac{1}{2}} D_1 C^{\frac{1}{2}} = (\idty - A)^{\frac{1}{2}} D_2 (\idty - C)^{\frac{1}{2}}.
\end{split}
\label{pos22}
\end{equation}

Free states are the fermionic version of classical Gaussians. For Gaussians, expectations of a random function multiplied by a Gaussian variable can be expressed as expectations of the derivative of the function with respect to the random variable. The following lemma provides such a formula in the fermionic context. Note that the classical second derivative becomes a combined commutation anti-commutation.

\begin{lemma}
\label{lemma_expect}
For any $Y \in \c A(\c H)$ and $\varphi \in \c H$, we have
\begin{equation}
\omega_Q (a^{*}(\varphi) Y a(\varphi)) = \omega_Q(a^*(\varphi) a(\varphi)) \omega_Q(Y) + \omega_Q \bigl( \bigl\{ a(\varphi) \,,\, \bigl[ a^*(\varphi) \,,\, Y \bigr] \bigr\} \bigr)
\end{equation}
\end{lemma}

\begin{proof}
We may limit ourselves to gauge-invariant $Y$ due to the gauge-invarian\-ce of the free state. Since we can approximate $Y$ by linear combinations of gauge-invariant monomials in $\c A(\c H)$, it suffices to show the lemma for $Y=a^*(\psi_1) \cdots a^*(\psi_n) a(\eta_n) \cdots a(\eta_1)$. For such $Y$, using the fact that $\omega_Q$ is free, the expression $\omega_Q(a^{*}(\varphi) Y a(\varphi))$ expands to
\begin{equation}
\begin{split}
&\omega_Q \bigl( a^*(\varphi) a(\varphi) \bigr)\, \omega_Q \bigl( a^*(\psi_1) \cdots a^*(\psi_n) a(\eta_n) \cdots a(\eta_1) \bigr) \\
&\qquad+ \sum_{k,\ell} \varepsilon_{k,\ell}\, \omega_Q \bigl( a^*(\varphi) a(\eta_\ell) \bigr)\, \omega_Q \bigl( a^*(\psi_k) a(\varphi) \bigr) \\[-12pt] 
&\phantom{\qquad+ \sum_{k,\ell} \epsilon_{k,\ell}}\ \times \omega_Q \bigl( a^*(\psi_1) \cdots \widehat{a^*(\psi_k)} \cdots a^*(\psi_n) a(\eta_n) \cdots \widehat{a(\eta_\ell)} \cdots a(\eta_1) \bigr). 
\end{split}
\end{equation}
Here $\varepsilon_{k,\ell} = \pm1$, depending on the parity of the permutation needed to put the modes in the original order and $\widehat{a^*(\psi_k)}$ means that the factor $a^*(\psi_k)$ is removed from the product $a^*(\psi_1) \cdots a^*(\psi_n)$.

We now compute by repeated application of eq.~(\ref{CAR})
\begin{align}
&a^*(Q\varphi) a(\eta_n) \cdots a(\eta_1) \\
&\quad= \< \eta_n \,,\, Q \varphi \> a(\eta_{n-1}) \cdots a(\eta_1) - a(\eta_n) a^*(Q\varphi) a(\eta_{n-1}) \cdots a(\eta_1) \\
&\quad= \sum_\ell \varepsilon_{\ell}\, \omega_Q\bigl( a^*(\varphi) a(\eta_\ell) \bigr)\, a(\eta_n) \cdots \widehat{a(\eta_\ell)} \cdots a(\eta_1) \pm a(\eta_n) \cdots a(\eta_1) a^*(Q\varphi),
\end{align}  
with the upper sign for $n$ even and the lower sign for $n$ odd, therefore
\begin{equation}
\sum_\ell \varepsilon_{\ell}\, \omega_Q\bigl( a^*(\varphi) a(\eta_\ell) \bigr)\, a(\eta_n) \cdots \widehat{a(\eta_\ell)} \cdots a(\eta_1) = \bigl[ a^*(Q\varphi) \,,\, a(\eta_n) \cdots a(\eta_1) \bigr]_\mp.
\end{equation}
Using this relation, its conjugate and the anti-commutation relations~(\ref{CAR}), we get the desired result for gauge-invariant monomials and hence for all gauge-invariant elements
\begin{align}
\begin{split}
&\sum_{k,\ell} \varepsilon_{k,\ell}\, \omega_Q \bigl( a^*(\varphi) a(\eta_\ell) \bigr)\, \omega_Q \bigl( a^*(\psi_k) a(\varphi) \bigr) \\[-12pt] 
&\phantom{\qquad+ \sum_{k,\ell} \epsilon_{k,\ell}}\ \times  a^*(\psi_1) \cdots \widehat{a^*(\psi_k)} \cdots a^*(\psi_n) a(\eta_n) \cdots \widehat{a(\eta_\ell)} \cdots a(\eta_1) \\[-12pt] 
&= - \sum_{\ell} \varepsilon_{\ell} \, \omega_Q \bigl( a^*(\varphi) a(\eta_\ell) \bigr)\\[-12pt] 
&\phantom{\qquad+ \sum_{k,\ell} \epsilon_{k,\ell}}\ \times \mp \bigl[ a(Q\varphi) \,,\, a^*(\psi_1) \cdots a^*(\psi_n) \bigr]_\mp a(\eta_n) \cdots \widehat{a(\eta_\ell)} \cdots a(\eta_1) \\[-12pt] 
&= \pm \sum_{\ell}  \varepsilon_{\ell} \, \omega_Q \bigl( a^*(\varphi) a(\eta_\ell) \bigr)\\[-12pt] 
&\phantom{\qquad+ \sum_{k,\ell} \epsilon_{k,\ell}}\ \times \bigl\{ a(Q\varphi) \,,\, a^*(\psi_1) \cdots a^*(\psi_n) a(\eta_n)  \cdots \widehat{a(\eta_\ell)} \cdots a(\eta_1) \bigr\} \\[-12pt]
&=\pm \bigl\{ a(Q\varphi) \,, a^*(\psi_1) \cdots a^*(\psi_n)  \bigl[ a^*(Q\varphi), a(\eta_n) \cdots a(\eta_1) \bigr]_\mp \bigr\} \\
&= \bigl\{ a(Q\varphi) \,, \bigl[ a^*(Q\varphi), a^*(\psi_1) \cdots a^*(\psi_n) a(\eta_n) \cdots a(\eta_1) \bigr] \bigr\}.
\end{split}
\end{align}
\end{proof}

The following proposition bounds the two-point correlations of conditional states. To show this we rely on the equilibrium properties of free states. An equilibrium state $\omega_\beta$ on a C*-algebra $\c A$ is linked to a dynamics in Heisenberg picture through the KMS-condition. Let $\{ \alpha_t \mid t \in \Rl \}$ be a continuous group of automorphisms of $\c A$, then $\omega_\beta$ is an $\alpha$-KMS-state at inverse temperature $\beta > 0$ if there exists for any pair of observables $x,y \in \c A$ a function
\begin{equation}
z \in \Cx \mapsto F_{x,y}(z) \in \Cx
\end{equation}
that is analytic inside the strip $0 < \Im\g m z < \beta$, that extends continuously to the closure of the strip, and such that
\begin{equation}
F_{x,y}(t) = \omega_\beta(\alpha_t(x) y) \enskip\text{and}\enskip F_{x,y}(t+i\beta) = \omega_\beta(y \alpha_t(x)),\enskip t \in \Rl.
\end{equation}
It is straightforward to check that the KMS-states on a finite dimensional full quantum system precisely coincide with the canonical Gibbs states.

Let $0 < Q < \idty$ whereby we mean that for $0 \ne \varphi$
\begin{equation}
0 < \< \varphi \,,\, Q \varphi \> \enskip\text{and}\enskip 0 < \< \varphi \,,\, (\idty - Q) \varphi \>. 
\end{equation}  
The state $\omega_Q$ is then the unique $\alpha$-KMS-state on $\c A(\c H)$ at inverse temperature $\beta=1$ where $\alpha$ is the strongly continuous one-parameter group of automorphisms~\cite{RST}
\begin{equation}
\alpha_t \bigl( a^*(\varphi) \bigr) = a^*\bigl( \r e^{ith} \varphi \bigr),\enskip t \in \Rl
\end{equation}
with
\begin{equation}
h = \ln (\idty - Q) - \ln Q.
\label{ham}
\end{equation}

\begin{proposition}
\label{prop_corr}
With the assumptions and notations of above there exists for any positive, gauge-invariant  $Y \in \c A(\c H_2)$ with $\omega_Q(Y) = 1$ a bounded operator $\tilde A$ on $\c H_1$ such that
\begin{equation}
\omega_Q \bigl( a^*(\varphi) a(\psi) Y \bigr) = \< \psi \,,\, \tilde A \varphi \>,\enskip \varphi,\psi \in \c H_1 
\end{equation} 
and
\begin{equation}
A - B C^{-1} B^* \le \tilde A \le A + B (\idty - C)^{-1} B^*.
\end{equation}
\end{proposition}

\begin{proof}
Since $Y$ commutes with $a(\psi)$ and $\omega_Q(Y)=1$, we can use lemma~\ref{lemma_expect} to get
\begin{align}
\< \psi \,,\, \tilde A \varphi \> &= \omega_Q (a^*(\varphi)Ya(\varphi)) \\
 &= \< \varphi \,,\, A\varphi \> + \omega_C \bigl( \bigl\{ a(B^*\varphi) \,,\, \bigl[ a^*(B^*\varphi) \,,\, Y \bigr] \bigr\} \bigr).
\label{id}
\end{align}

The next step consists in rewriting this identity in such a way that we can use the information $\omega_C(Y) = 1$. This can be achieved through the KMS-condition. Using an approximation argument we may assume that $a^*(B^* \varphi)$ and $Y$ are analytic elements for the automorphism group $\alpha$ of $\c A(\c H_2)$ defined by $C$. For an analytic element $x$ we have
\begin{align}
&\bigl( \alpha_z(x) \bigr)^* = \alpha_{\overline z}(x^*),\enskip z \in \Cx\enskip \text{and} \\
&\omega_C(xY) = \omega_C \bigl(x Y^{\frac{1}{2}} Y^{\frac{1}{2}} \bigr) = \omega_C \bigl( \alpha_{-i}(Y^{\frac{1}{2}}) x Y^{\frac{1}{2}} \bigr) \nonumber \\
&\phantom{\omega_C(xY)}= \omega_C \bigl( \alpha_{-\frac{i}{2}}(Y^{\frac{1}{2}}) \alpha_{\frac{i}{2}}(x) \alpha_{\frac{i}{2}}(Y^{\frac{1}{2}}) \bigr),
\end{align}
and similarly
\begin{align}
&\omega_C(Yx)=\omega_C \bigl( \alpha_{-\frac{i}{2}}(Y^{\frac{1}{2}}) \alpha_{-\frac{i}{2}}(x) \alpha_{\frac{i}{2}}(Y^{\frac{1}{2}}) \bigr), \\
&\omega_C(xYy)=\omega_C \bigl( \alpha_{-\frac{i}{2}}(Y^{\frac{1}{2}}) \alpha_{-\frac{i}{2}}(y) \alpha_{\frac{i}{2}}(x) \alpha_{\frac{i}{2}}(Y^{\frac{1}{2}}) \bigr).
\end{align}
This allows us to rewrite~(\ref{id}) as
\begin{equation}
\< \varphi \,,\, \tilde A\varphi \> = \< \varphi \,,\, A\varphi \> + \omega_C \bigl( \alpha_{-\frac{i}{2}}(Y^{\frac{1}{2}}) u \alpha_{\frac{i}{2}}(Y^{\frac{1}{2}})\bigr)
\end{equation} 
with
\begin{equation}
u = \bigl\< B^* \varphi \,,\, \bigl( \idty + \r e^{-h} \bigr) B^* \varphi \bigr\> - a^* \bigl( \r e^{\frac{h}{2}} B^*\varphi + \r e^{-\frac{h}{2}} B^*\varphi \bigr) a \bigl( \r e^{\frac{h}{2}} B^*\varphi + \r e^{-\frac{h}{2}} B^*\varphi \bigr).
\end{equation}
Here $h$ is the single mode Hamiltonian as in~(\ref{ham}) replacing $Q$ by $C$. Using 
\begin{equation}
0 \le a^*(\zeta) a(\zeta) \le \norm \zeta^2 \idty
\end{equation} 
we obtain the statement of the proposition
\begin{equation}
A - B C^{-1} B^* \le \tilde A \le A + B (\idty - C)^{-1} B^*. 
\end{equation}
\end{proof}

Obviously, the two-point correlations of states in $S_1^{\text{free}}$ also satisfy these bounds. Since the two-point correlations of a free state $\omega_Q$ are encoded in its symbol $Q$, the operator $Q$ will satisfy the bounds given for $\tilde A$ in proposition \ref{prop_corr}. In proposition~\ref{pro2}, we show that the converse is also true, i.e. that every free state whose two-point correlations satisfy the given bounds is contained in the weak$^*$-closure of $S_1^{\text{free}}$.

To prove this statement, we use conditional states generated by an exponential element $Y$ in $\c A(\c H_2)$.

\begin{lemma}
\label{lemma_exp_cond}
If $\omega_Q$ is a free state on $\c A(\c H)$ with symbol $Q$ as in~(\ref{sym}) and $Y = E(L)/\omega_Q(E(L))$ is an exponential element in $\c A(\c H_2)$ with $L \ge 0$, then the conditional state $\tilde \omega: X \mapsto \omega_Q (XY)$ is a free state on $\c A(\c H_1)$ with symbol
\begin{equation}
\label{cond_symb}
\tilde A = A - B(L-\idty)(\idty -C+CL)^{-1}B^*.
\end{equation}
\end{lemma}

\begin{proof}
We calculate the expectation value of elements $X=E(K)$ with $K$ an operator on $\c H_1$ in the state $\tilde \omega$. Since these elements $E(K)$ span the gauge invariant CAR algebra~\cite{DFP}, these values determine the state $\tilde{\omega}$.

First we determine the normalization factor $\omega_Q(E(L))$ by using eq.~(\ref{gi-free})
\begin{align}
	\omega_Q (E(L)) &= \det( \idty - Q + Q ( \idty \oplus L)) \\
				& = \det \begin{pmatrix} \idty & -B + BL \\ 0 & \idty -C + CL \end{pmatrix} \\
				& = \det(\idty -C +CL).
\end{align}
Likewise, we have that
\begin{align}
\omega_Q(E(K)E(L))&=  \det( \idty - Q + Q ( K \oplus L)) \\
			& =  \det \begin{pmatrix} \idty - A + AK & -B + BL \\ -B^* + B^* K & \idty -C + CL \end{pmatrix} \\
			& = \det(\idty -C +CL) \det(\idty - \tilde{A} + \tilde{A}K)
\end{align}
with
\begin{equation}
\tilde{A}=A-B(L-\idty)(\idty -C +CL)^{-1} B^*.
\end{equation}
Hence, $\tilde{\omega}$ is a free state with symbol $\tilde{A}$
\begin{equation}
\tilde{\omega} (E(K)) = \frac{\omega_Q(E(K)E(L))}{\omega_Q (E(L))} =  \det(\idty - \tilde{A} + \tilde{A}K) = \omega_{\tilde{A}} (E(K)).
\end{equation}
\end{proof}

\begin{lemma}
\label{lem2}
Let $0 < \varepsilon < 1$ and let $\tilde A$ be an operator on $\c H_1$ such that $A - \tilde A$ is of finite rank and such that
\begin{equation}
A - (1-\varepsilon) B C^{-\frac{1}{2}} B^* \le \tilde A \le A + (1-\varepsilon) B (\idty - C)^{-\frac{1}{2}} B^*,
\label{lem2:1}
\end{equation}
then there exists a positive $Y \in \c A(\c H_2)$ such that
\begin{equation}
\omega_{\tilde A}(X) = \omega_Q(XY),\enskip X \in \c A(\c H_1).
\end{equation} 
\end{lemma}

\begin{proof}
We consider the set of operators
\begin{equation}
\label{finrank_symb}
\tilde A = A + BKB^*
\end{equation}
with $K$ a finite rank operator on $\c H_2$ such that $\tilde A$ satisfies the bounds~(\ref{lem2:1}). This is the case if
\begin{equation}
- (1-\varepsilon) C^{-1} \le K \le (1-\varepsilon) (\idty - C)^{-1}.
\end{equation}

Using lemma~\ref{lemma_exp_cond}, we obtain the free state with symbol $\tilde A$ as the conditional state $X \mapsto \omega_Q(XY)$ with 
\begin{equation}
Y = \frac{1}{\omega_Q(E(L))}\, E(L) \in \c A(\c H_2) 
\end{equation}
if we are able to find a positive operator $L$ on $\c H_2$ such that
\begin{equation}
K = (\idty-L)(\idty -C +CL)^{-1} \enskip\text{and}\enskip \idty - L \text{ finite rank}.
\end{equation}
Rewriting this in terms of a finite rank operator $N$, such that $L=\idty + N$, we have
\begin{equation}
K = -N(\idty + CN)^{-1}.
\end{equation}
If $\idty + CK$ is invertible, this equation is solved by
\begin{equation}
N=-K(\idty + CK)^{-1}.
\end{equation}

To show that $\idty + CK$ is invertible, assume that $\varphi \in \ker \bigl\{ \idty +KC \bigr\}$. This means that
\begin{equation}
\< C\varphi, \varphi \> + \<C\varphi, KC \varphi\> = 0
\end{equation}
and
\begin{equation}
0 \ge \< C\varphi, \varphi \> - \< C \varphi, (1-\varepsilon) C^{-1}C\varphi\> = \varepsilon \<\varphi, C\varphi\>.
\end{equation}
Hence $\ker\bigl\{\idty + KC \bigr\} = \bigl\{ 0 \bigr\}$. Therefore, as $CK$ is of finite rank, $\ran(\idty + CK) = \c H_2$. Furthermore $\ker\bigl\{\idty + CK \bigr\} = \bigl\{ 0 \bigr\}$ as well and so $\idty + CK$ is invertible.
\end{proof}

\begin{proposition}
\label{pro2}
The weak$^*$-closure of the set $\c S_1^{\text{free}}$ of conditioned free states on $\c A (\c H_1)$ coincides with the set of free states on $\c A(\c H_1)$ whose symbols $\tilde A$ satisfy 
\begin{equation}
A - B C^{-\frac{1}{2}} B^* \le \tilde A \le A + B (\idty - C)^{-\frac{1}{2}} B^*.
\label{pro2:1}
\end{equation}
\end{proposition}

\begin{proof}
For free states, weak$^*$-convergence is equivalent to weak convergence of their symbols. The proof then immediately follows from proposition~\ref{prop_corr} and lemma~\ref{lem2}.
\end{proof} 

\begin{comment}
Of course, symbols $\tilde A$ that are weak limits of the symbols used in the proof of lemma~\ref{lem2} satisfy the inequalities of proposition~\ref{prop_corr}. This can be shown directly. To do so we need that
\begin{equation}
-C^{-1} \le (\idty-L)(\idty -C +CL)^{-1} \le (\idty - C)^{-1}
\end{equation}
for any $L \ge 0$.
By multiplying on the right by $(\idty -C + CL)$ and on the right by its conjugate, these inequalities are equivalent with
\begin{equation}
-C^{-1} + \idty -2L - LCL \le -L -LCL \le LC(\idty -C)^{-1}CL
\end{equation}
which are satisfied due to the fact that $0 \le L$ and $0 \le C \le \idty$.
\end{comment}

There is a notion of gauge-invariant, free, completely positive, identity preserving maps between CAR algebras that naturally extends that of free states. Such a map $\Gamma: \c A(\c H) \to \c A(\c K)$ is determined by operators
\begin{equation}
R: \c H \to \c K \enskip\text{and}\enskip S: \c H \to \c H
\end{equation}
that satisfy
\begin{equation}
0 \le S \le \idty - R^*R.
\end{equation}
The action of the map on a gauge-invariant monomial of order two is given by
\begin{equation}
\Gamma(a^*(\varphi) a(\psi)) = a^*(R\varphi) a(R\psi) + \< \psi \,,\, S\,\varphi \>,\enskip \varphi,\psi \in \c H.
\label{qfmap}
\end{equation}
Moreover, the pull-back $\omega_Q \circ \Gamma$ of a free state on $\c A(\c K)$ is a free state on $\c A(\c H)$. For more details, we refer to~\cite{DFP}.

As in eq.~(\ref{deco2}), we can write the free conditional states as generalized subsystems, using a free completely positive map to a suitable operator algebra and free states on the target algebra.

\begin{proposition}
\label{pro3}
There exists a unique, free, minimal, identity preserving, completely positive map $\Gamma$ such that the weak$^*$-closure of $\c S^{\text{free}}_1$ is the pull-back of the free states by $\Gamma$. 
\end{proposition}

\begin{proof}
Let $\c K = \overline{\ran(B)} \subset \c H_1$. We construct operators 
\begin{equation}
R: \c H_1 \to \c K \enskip\text{and}\enskip S: \c H_1 \to \c H_1 
\end{equation}
such that
\begin{equation}
0 \le S \le \idty - R^*R. 
\end{equation}
These operators define a completely positive, free, identity preserving map $\Gamma$ from $\c A(\c H_1)$ to $\c A(\c K)$ as in~(\ref{qfmap}). The pull-back of the free states on $\c A(\c K)$ consists of the free states on $\c A(\c H_1)$ with symbols
\begin{equation}
\{ \tilde A = R^*T\,R + S \mid 0 \le T \le \idty \}.
\label{pro3:2}
\end{equation}
We need to show that the set~(\ref{pro3:2}) coincides with~(\ref{pro2:1}). This is the case if and only if 
\begin{equation}
R = U\, \sqrt{B\, C^{-\frac{1}{2}} B^* + B\, (\idty - C)^{-\frac{1}{2}} B^*} \enskip\text{and}\enskip S = A - B\, C^{-\frac{1}{2}} B^*.
\end{equation} 
In this expression $U$ is an arbitrary unitary on $\c K$.
\end{proof}

\noindent
\textbf{Acknowledgements} \\
This work is partially funded by the Belgian Interuniversity Attraction Poles Programme P6/02.

\end{document}